\documentclass[runningheads]{llncs}
\usepackage{graphicx}

\usepackage{booktabs} 

\usepackage{hyperref}
\usepackage{url}
\usepackage[english]{babel}
\usepackage{algorithm}
\usepackage[noend]{algpseudocode}
\usepackage{amssymb,amsmath}
\usepackage{subcaption}
\usepackage{graphicx}
\usepackage{comment}
\usepackage{thmtools, thm-restate}
\usepackage{multirow}

\usepackage{xcolor} 
\definecolor{ocre}{RGB}{243,102,25} 

\usepackage{avant} 
\usepackage{mathptmx} 

\usepackage{microtype} 
\usepackage[utf8x]{inputenc} 
\usepackage[T1]{fontenc} 

\usepackage{tikz}
\usetikzlibrary{shapes}
\RequirePackage{ifthen}
\RequirePackage{amsmath}
\RequirePackage{amssymb}
\RequirePackage{xspace}
\RequirePackage{graphics}
\usepackage{color}
\RequirePackage{textcomp}
\usepackage{keyval}
\usepackage{xspace}
\usepackage{mathrsfs}
\usepackage{enumitem}

\newcommand{\sayan}[1]{\textcolor{blue}{#1}}

\newcommand{\reals}{{\mathbb{R}}}

\newcommand{\dom}{\relax\ifmmode {\sf dom} \else ${\sf dom}$\fi}

\newcommand{\Reach}{\mathit{Reach}} 

\newcommand{\Rtube}{\mathtt{ReachTb}}
\newcommand{\ARtube}{\mathtt{ReachTb}}
\newcommand{\AReach}{\mathtt{ReachTb}} 
\newcommand{\argmax}{\sf{argmax}} 

\newcommand{\node}{\mathit{node}}

\newcommand{\reachset}{\mathtt{Reach}}

\newcommand{\concat}{\mathbin{^{\frown}}} 

\newcommand{\tubecache}{\mathit{tubecache}}
\newcommand{\safetycache}{\mathit{safetycache}}
\newcommand{\symcompute}{\mathtt{symComputeReachtube}} 
\newcommand{\symsafety}{\mathtt{symSafetyVerif}}
\newcommand{\tubecompute}{\mathtt{computeReachtube}}
\newcommand{\safe}{\mathit{safe}}
\newcommand{\unsafe}{\mathit{unsafe}}
\newcommand{\state}{{\bf x}}
\newcommand{\statey}{{\bf y}}
\newcommand{\rtube}{\mathit{rtube}}
\newcommand{\ourtool}{\sf{CacheReach}}
\newcommand{\gettube}{\sf{getTube}}
\newcommand{\multiagentverif}{\sf{symMultiVerif}}
\newcommand{\dynamicsafety}{\sf{checkDynamicSafety}}
\newcommand{\modetube}{\mathit{mtube}}
\newcommand{\agenttube}{\mathit{atube}}
\newcommand{\initset}{\mathit{initset}}
\newcommand{\guardintersect}{\mathit{guardIntersect}}
\newcommand{\guard}{\mathit{guard}}
\newcommand{\ltube}{\mathit{len}}
\newcommand{\relevant}{\mathit{O}}
\newcommand{\ftime}{\mathit{ftime}}
\newcommand{\etime}{\mathit{etime}}

\usepackage{setspace}
\usepackage{listings}
\usepackage{hyperref}
\hypersetup{
	pdftitle={Modeling and Verification for CPS},
	pdfauthor={Sayan Mitra},
	colorlinks=true,
	citecolor={blue},
	linkcolor = {blue},
	pagecolor={blue},
	backref={true},
	bookmarks=true,
	bookmarksopen=false,
	bookmarksnumbered=true
}

\usepackage[normalem]{ulem}

\usepackage{tikz}
\usetikzlibrary{fit}
\usepackage{twoopt}

\newsavebox\IBoxA \newsavebox\IBoxB \newlength\IHeight
\newcommand\TwoFig[6]{
	\sbox\IBoxA{\includegraphics[width=0.45\textwidth]{#1}}
	\sbox\IBoxB{\includegraphics[width=0.45\textwidth]{#4}}%
	\ifdim\ht\IBoxA>\ht\IBoxB
	\setlength\IHeight{\ht\IBoxB}%
	\else\setlength\IHeight{\ht\IBoxA}\fi
	\begin{figure}[!htb]
		\minipage[t]{0.45\textwidth}\centering
		\includegraphics[height=\IHeight]{#1}
		\caption{#2}\label{#3}
		\endminipage\hfill
		\minipage[t]{0.45\textwidth}\centering
		\includegraphics[height=\IHeight]{#4}
		\caption{#5}\label{#6}
		\endminipage 
	\end{figure}%
}

\newcommandtwoopt\Textbox[5][2.5cm][2cm]{%
	\begin{tikzpicture}[remember picture,overlay]
	\coordinate (aux) at ([xshift=#1]#4);
	\node[inner ysep=3pt,yshift=0.6ex,draw=gray,thick,
	fit=(#3) (aux),baseline] 
	(box) {};
	\node[text width=#2,anchor=north east,
	font=\sffamily\footnotesize,align=right] 
	at (box.north east) {#5};
	\end{tikzpicture}%
}

\titlerunning{Multi-Agent Safety Verification using	 Symmetry Transformations}
\authorrunning{Sibai, Mokhlesi, Fan, and Mitra}

\begin{document}
\title{Multi-Agent Safety Verification using Symmetry Transformations}
%
%
\author{Hussein Sibai \and
Navid Mokhlesi\and
Chuchu Fan \and
Sayan Mitra
\institute{University of Illinois, Urbana IL 61801, USA}
\email{\{sibai2,navidm2,cfan10,mitras\}@illinois.edu}
}
%
%

%
\maketitle              
\begin{abstract}
We show that symmetry transformations and caching can enable scalable, and possibly unbounded, verification of multi-agent systems. 
Symmetry transformations map solutions and to other solutions. We show that this property can be used to transform cached reachsets to compute new reachsets, for hybrid and multi-agent models. 
We develop a notion of {\em virtual system\/} which  define symmetry transformations for a broad class of agent models that visit waypoint sequences. 
Using this notion of virtual system, we present a prototype tool $\ourtool$  that builds a cache of reachtubes for this system, in a way that is   agnostic of the representation of the reachsets and the reachability analysis subroutine used. 
Our experimental evaluation of $\ourtool$ shows up to 66\% savings in safety verification computation time on multi-agent systems with 3-dimensional linear and 4-dimensional nonlinear fixed-wing aircraft models following sequences of waypoints. 
These savings and our theoretical results illustrate the potential benefits of using symmetry-based caching in the safety verification of multi-agent systems.
\end{abstract}

\section{Introduction}
\label{sec:intro}

As the cornerstone for safety verification of dynamical and hybrid systems, reachability analysis has attracted  attention and has delivered automatic analysis  of automotive, aerospace, and medical applications~\cite{Online-AlthoffD14,KushnerBCFMS19,FanQM18,DuggiralaFM015}.
Notable advances from the last few years include the development of the generalized star data-structure~\cite{Star:Duggirala016} and the HyLaa tool~\cite{bak2017hylaa} that have massive linear models~\cite{BakTJ19}; Taylor model based reachability analysis algorithms for nonlinear systems and their implementations in Flow*~\cite{flow}; and  a simulation-based algorithm that guarantees locally optimal precision~\cite{FanKJM16:Emsoft}.

Exact symbolic reachability analysis of nonlinear models is generally hard. One prominent approach is based on generalizing individual behaviors or simulations to cover a whole set of behaviors. The idea was pioneered in~\cite{donze2007systematic} and implemented in Breach~\cite{breachAD} with sound generalization guarantees for linear models based on {\em sensitivity analysis}. Subsequently the idea has been extended to nonlinear, hybrid, and black-box systems and implemented in tools like C2E2 and DryVR~\cite{DMVemsoft2013,FanQM0D16,FanQM18,FanM15:ATVA}. 

In all of the above, a single behavior $\xi$ of the system from an  initial state, is generalized to a {\em compact set\/}  of {\em neighboring\/} behaviors that contains all the behaviors   starting from a small neighborhood around the initial state of $\xi$. Thus, the computed neighboring set of behaviors always contains $\xi$ and its size is determined by the algorithms for sensitivity analysis. 
In contrast, the type of generalization we pursue here  uses {\em symmetry transforms\/} on the state space. Given a group $\Gamma$ of operators on the state space, and a  single behavior $\xi$, we can generalize it to $\gamma(\xi)$, for each $\gamma \in \Gamma$. 
The transformation can be applied to sets of behaviors symbolically. Not only can this type of generalization work in conjunction with sensitivity analysis, it captures structural properties of the system that make its behavior similar for significantly different states that sensitivity analysis would not capture. 
%

In our recent work~\cite{Sibai:ATVA2019}, we showed how this symmetry transforms can be used to produce new reach sets from other previously computed reach sets for non-parameterized dynamical systems. 
In this paper, we introduce the use of symmetry transforms of parameterized dynamical systems in their safety verification. We present an algorithm $\symcompute$ (Algorithm~\ref{code:sym_annotations}) which cache and reuse reach sets to generate new ones, avoiding repeating expensive computations.
We show how an infinite number of reach sets can be obtained by transforming a single one using symmetry transforms (Corollary~\ref{cor:vir_real_eq_inter}). Building on it, we provide unbounded time safety guarantees using finite cached safety checking results (Theorem~\ref{thm:unboundedsafety}).






The key contributions of this paper are as follows.

First, we show how symmetry transformations for parameterized dynamical systems can be used to compute reachable states (Theorem~\ref{thm:tube_trans_input}). Going well beyond the previous theory~\cite{Sibai:ATVA2019}, this enables {\em cached reach tubes\/} to be reused for verification across different modes and across multiple agents. 

We develop a notion of {\em virtual system\/} (Section~\ref{sec:virtual}) which automatically define symmetry transformations for a broad swathe of hybrid and dynamical systems modeling agents visiting a sequence of waypoints (see Theorem~\ref{thm:virtual_rep} and Examples~\ref{sec:virtual_example} and~\ref{sec:virtual_tube_example}). That is, reachability analysis of a multi-agent system (with possibly different dynamics and different parameters) can be performed in a common transformed virtual coordinate system, and thus, increasing the possibility of reuse. 
We show how this principle can make it possible to verify systems over unbounded time and with unbounded number of agents (Theorem~\ref{thm:unboundedsafety}), provided we infer that no new unproven scenarios appear for the virtual system. 

We present a prototype implementation of a tool that uses $\symcompute$. We name it $\ourtool$. It builds a cache of reach tubes for the transformed virtual system, from different sets of initial states. In performing reachability analysis of a multi-agent hybrid or dynamical system, for each agent and each mode, it: 
1. transforms the initial set  $X$ to  the virtual coordinates to get $\gamma(X)$. 
2. If the transformed set $\gamma(X)$ has already been stored in the cache, then it simply extracts the cached reach tube and $\gamma^{-1}$ transforms it to get the actual reach set.  
3. Otherwise, it computes the reach set from $\gamma(X)$ and caches it. 
Our algorithm $\symcompute$ and its implementation in $\ourtool$ are agnostic of the representation of the reach sets and the reachability analysis subroutine, and therefore, any of the ever-improving libraries can be plugged-in for step 3.

Our experimental evaluation of $\ourtool$ shows safety verification computation time savings of up to 66\% on scenarios with multiple agents with 3d linear dynamics and 4d fixed-wing aircraft nonlinear model following sequences of waypoints. 
These savings illustrate the potential benefits of using symmetry transformation-based caching in the safety verification of multi-agent systems.

\section{Model and problem statement}
\label{sec:model}
\paragraph{Notations.}
We denote by $\mathbb{N}$, $\mathbb{R}$, and $\mathbb{R}^{\geq 0}$ the sets of natural numbers, real numbers and non-negative reals, respectively. Given a finite set $S$, its cardinality is denoted by $|S|$. Given $N \in \mathbb{N}$, we denote by $[N]$ the set $\{1, \dots, N\}$. Given a vector $v \in \mathbb{R}^n$ and a set $L \subseteq [n]$, we denote the projection of $v$ to the indices in $L$ by $x[L]$. We define an $n$-dimensional hyper-rectangle by a 2d-array specifying its bottom-left and upper-right vertices in $\mathbb{R}^n$. We denote the projection of a hyperrectange $H$ on the set of dimensions $L$ by $H \downarrow_L$. 
Given a function $\gamma: \mathbb{R}^k \rightarrow \mathbb{R}^k$ and a set $S \subseteq \mathbb{R}^k$, we abuse notation and define 
$\gamma(S) = \{ \gamma(x)\ |\ x \in S\}$. 
Moreover, given $S \in 2^{\mathbb{R}^k} \times \mathbb{R}^{\geq 0}$, we define $\gamma(S) = \{ (\gamma(X), t)\ |\ (X,t) \in S \}$.

\subsection{Agent mode dynamics}
\label{sec:agent-model}

In this section, we define the syntax and semantics of the model that determines the dynamics of an agent. We present the syntax first.

\begin{definition}[syntax]
\label{def:model_syn}
The agent dynamics are defined by a tuple $A = \langle S, P, f \rangle$, where $S \subseteq \mathbb{R}^n$ is its state space, $P \subseteq \mathbb{R}^m$ is its parameter or mode space, and the dynamic function $f: S \times P \rightarrow S$ that is Lipschitz in the first argument.
\end{definition}

The semantics of an agent dynamics is defined by trajectories, which describe the evolution of states over time.

\begin{definition}[semantics]
\label{def:model_sem}
For a given agent $A = \langle S, P, f \rangle$, we call a function $\xi: S \times P \times \mathbb{R}^{\geq 0} \rightarrow S$ a {\em trajectory} if $\xi$ is differentiable in its third argument, and given an initial state $\state_0 \in S$ and a mode $p \in P$, $\xi(\state_0,p,0) = \state_0$ and for all $t > 0$,
\begin{align}
\label{sys:input}
\frac{d \xi}{d t}(\state_0,p,t) = f(\xi(\state_0,p,t), p).
\end{align}
We say that $\xi(\state_0,p,t)$ is the state of $A$ at time $t$ when it starts from $\state_0$ in mode $p$. 

%
\end{definition}
Given an initial state $\state_0 \in S$ and mode $p \in P$, the trajectory $\xi(\state_0, p,\cdot)$ is the unique solution of the ordinary differential equation (ODE) (\ref{sys:input}) since $f$ is Lipschitz continuous.

Given a compact initial set $K \subseteq S$, a parameter $p \in P$, the set of {\em reachable states\/} of $A$ over a  time interval $[\ftime,\etime]$ is defined as 
\begin{align}
\label{eq:reach}	
\reachset(K, p, [\ftime,\etime]) = \{\state \in S \ |\ \exists \state_0 \in K, t \in [\ftime, \etime], \state = \xi(\state_0, p, t) \}.
\end{align}
We let $\reachset(K, p, t)$ denote the set of reachable states at time $t$.
With unbounded time intervals, the set is denoted by $\reachset(K,p,[\ftime,\infty))$.
%

The {\em bounded time safety verification}  problem is to check if any state reachable by $A$ for a given $p \in P$ within the time bound is unsafe. That is, given a time bound $T > 0$, $p \in P$, and an unsafe set $U \subseteq S$, we want to check whether 
\[ \reachset(K,p,[0,T]) \cap U = \emptyset\] 
If the unsafe set is time dependent, i.e., $U \subseteq S \times \mathbb{R}^{\geq 0}$, then we need to use $\reachset(K,p,t)$ instead of $\reachset(K,p,[0,T])$ for checking safety.

\subsection{Reachtubes}
Computing reachsets exactly is theoretically hard \cite{HENZINGER199894}.
There are many reachability analysis tools~\cite{CAS13,Althoff2015a,bak2017hylaa} that can compute bounded-time over-approximations of the reachsets. Generally, given an initial set $K$ for a set of ODEs, these tools can return a sequence of sets that contain the exact reachset over small time intervals. Motived by this, we define reachtubes as over-approximations of exact reachsets:
\begin{definition}
For a given agent $A = \langle S, P, f \rangle$, an initial set $K \subseteq T$, a mode $p \in P$, and a time interval $ [\ftime,\etime]$,
a $(K,p,[\ftime,\etime])$-reachtube $\Rtube(K, p, [0,T])$ is  a sequence $\{(X_i,[\tau_{i-1},\tau_{i}])\}_{i=1}^{j}$ such that
$\reachset(K, p, [\tau_{i-1},\tau_{i}]) \subseteq X_i$, and $\tau_0 = \ftime < \tau_1 < \dots < \tau_j = \etime$.
Without loss of generality, we assume equal separation between the time points, i.e. $\exists\ \tau_s > 0,  \forall i \in [j], \tau_{i} - \tau_{i-1} = \tau_s$. 
\end{definition}
For a given $(K,p,[\ftime,\etime])$-reachtube $\rtube$, we denote its parameters by $\rtube.K$, $\rtube.p$, $\rtube.\mathit{\ftime}$, and $\rtube.\mathit{\etime}$, respectively, and its cardinality by $\rtube.\ltube$.

  We define union, truncation, concatenate, and time-shift operators on reachtubes. Fix $\rtube_1 = \{(X_{i,1}, [\tau_{i-1,1}, \tau_{i,1}])\}_{i=1}^{j_1}$ and $\rtube_2 = \{(X_{i,2}, [\tau_{i-1,2},\tau_{i,2}])\}_{i=1}^{j_2}$ to be two reachtubes, where $j_1 = \rtube_1.\ltube$ and $j_2 = \rtube_2.\ltube$. Assume that $\tau_{i,1} = \tau_{i,2}$ for all $i \in \min(j_1, j_2)$, we say they are {\em time-aligned}. Without loss of generality, assume that $j_1 \leq j_2$.  The operators are defined as follows:
\begin{itemize}
	\item {\em union}: $\rtube_1 \cup \rtube_2 = \{ (X_{i,1} \cup X_{i,2},   [\tau_{i-1,1},\tau_{i,1})]\}_{i=1}^{j_1} \cup \{X_{i,2}\}_{i=j_1 + 1}^{j_2}$.  
	\item $\mathit{timeShift}(\rtube_1, t_s) = \{(X_{i,1}, [t_s + \tau_{i-1,1}, t_s + \tau_{i,1}])\}_{i=1}^{j_1}$.
	\item {\em concatenation}: $\rtube_1 \concat \rtube_2  = \{\rtube_1 \cup \mathit{timeShift}(\rtube_2, \tau_{j_1,1}) \}$.
	\item $\mathit{truncate}(\rtube_1, t_c) = \{(X_{i,1},[\tau_{i-1,1},\tau_{i,1}] )\}_{i=1}^{k}$, where $\tau_{k,1} \geq t_c$ and $\tau_{k-1,1} < t_c$.
\end{itemize}

A {\em simulation\/} of system~(\ref{sys:input}) is an approximate reachtube with $X_0$ being a singleton state $x_0 \in K$. That is, a simulation is a representation of $\xi(x_0,p, \cdot)$. 
Several numerical solvers can compute such simulations as VNODE-LP\footnote{\url{http://www.cas.mcmaster.ca/~nedialk/vnodelp/}} and CAPD Dyn-Sys library \footnote{\url{http://capd.sourceforge.net/capdDynSys/docs/html/odes_rigorous.html}}. 
\begin{example}[Fixed-wing aircraft following a single waypoint]
\label{sec:single_drone_example}
Consider an agent with state space $S = \mathbb{R}^4$, parameter space $P = \mathbb{R}^4$, and $f: S \times P \rightarrow S$ defined as follows: for any $\state \in S$ and $p \in P$,
\[
f(\state, p) = [\frac{T_c - c_{d1}\state[0]^2}{m}, \frac{g}{\state[0]} \sin \phi,  \state[0] \cos \state[1], \state[0] \sin \state[1]], 
\]
where  $T_c = k_1 m (v_c - \state[0])$, $\phi = k_2 \frac{v_c}{g} (\psi_c - \state[1])$, $\psi_c = \arctan(\frac{\state[2] - p[2]}{\state[3] - p[3]})$, and $k_1, k_2, m, g, c_{d1}$, and $v_c$ are positive constants. The agent models a fixed-wing aircraft starting from a waypoint and following another in the 2D plane: $\state[0]$ is its speed, $\state[1]$ is its heading angle, $(\state[2], \state[3])$ is its position in the plane, $[p(0), p(1)]$ is the position of the source waypoint, and $(p[2],p[3])$ is the position of the  destination one. Note that the source waypoint does not affect the dynamics, but will be useful later in the paper.
\end{example}

\section{Symmetry and Equivariant Dynamical Systems}
\label{sec:symdef}

Symmetry plays a fundamental rule in the analysis of dynamical and control systems. It has been used for studying stability of feedback systems\cite{symmetryandstabilityfeedback}, designing observers \cite{bonnabel2008symmetry} and  controllers~\cite{controlledsymmetries_passivewalking}, and analyzing neural networks \cite{Gerard2006NeuronalNA}. In this section, we present definitions for symmetry transforms and their implications on systems that posses them.

\subsection{Symmetry of systems with inputs}
 
In the following, symmetry transformations are defined by the ability of computing new solutions of~(\ref{sys:input}) using already computed ones. First, let $\Gamma$ be a group of smooth maps acting on $S$.

\begin{definition}[Definition 2 in \cite{russo2011symmetries}]
	\label{def:symmetry}
 We say that $\gamma \in \Gamma$ is a symmetry of (\ref{sys:input}) if for any solution $\xi(x_0,p, \cdot)$, $\gamma (\xi(\state_0,p, \cdot))$ is also a solution.
\end{definition} 

	Using $\gamma$-symmetry, we can get a new trajectory without simulating the system but instead by just transforming the entire old trajectory using $\gamma(\cdot)$.

For a linear dynamical system $\dot{x} = A x$, it is easy to see that any matrix $B$ that commutes with $A$ is a symmetry for the system. Consider the dynamics of an agent to be $f(\state,p) = A\state$. We can write down the closed-form solution of the trajectories to be $\xi(\state_0,p,t) = \state_0 e^{At}$. If $\gamma(\state) = B \state, \forall \state \in S$ and $AB=BA$, then $\xi(B \state_0,p,t) = B \state_0 e^{At} = B \xi(\state_0,p,t)$.  Therefore, $B$ is a symmetry for the system if $A B = B A$.

In the following definition we characterize the conditions under which a transformation is a symmetry of (\ref{sys:input}).
\begin{definition}
	\label{def:equivariance_input}
	The dynamic function $f: S \times P \rightarrow S$ is said to be $\Gamma$-equivariant if for any $\gamma \in \Gamma$, there exists $\rho: P \rightarrow P$ such that for all $\state \in S$, $\frac{\partial \gamma}{\partial \state} f(\state, p) = f(\gamma(\state), \rho_\gamma(p))$.
\end{definition}

The following theorem shows that it is enough to check the condition in Definition~\ref{def:equivariance_input} to prove that a transformation is a symmetry of~(\ref{sys:input}). 

\begin{theorem}[part of Theorem 10 in \cite{russo2011symmetries}]
	\label{thm:sol_transform_input_nonlinear}
	If $f$ is $\Gamma$-equivariant, then for any $\gamma \in \Gamma$, if $\xi(\state_0, p,\cdot)$ is a solution of (\ref{sys:input}), then so is $\gamma (\xi(\state_0,p,\cdot))$, which is equal to $\xi(\gamma(\state_0),\rho_\gamma(p),\cdot)$, where $\rho$ is the transformation associated with $\gamma$ according to Definition~\ref{def:equivariance_input}.
\end{theorem}
\begin{proof}
	Let ${\bf y} = \gamma(\state)$, then $\dot{{\bf y}} = \frac{\partial \gamma}{\partial \state} (\dot{\state}) = \frac{\partial \gamma}{\partial \state} (f(\state,p)) = f(\gamma(\state), \rho(p)) = f({\bf y}, \rho(p))$. The second equality is a result of the derivative chain rule. The $3^{\mathit{rd}}$ equality uses Definition~\ref{def:equivariance_input}.
\end{proof}

\begin{remark}
	Note that if $\gamma$ in Theorem~\ref{thm:sol_transform_input_nonlinear} is linear, the condition in Definition~\ref{def:equivariance_input} for a map $\gamma$ to be a symmetry becomes $\gamma(f(\state,p)) = f(\gamma(\state),\rho(p))$. 
\end{remark}

Note that the commutativity condition with the vector field of the ODE in Definition~\ref{def:equivariance_input} is similar to the properties of Lyapanov functions in the sense that it proves a property about the solutions without inspecting them individually. 

\begin{example}[Fixed-wing aircraft coordinate transformation symmetry]
\label{sec:transformation_example}
Consider the fixed-wing aircraft model presented in Section~\ref{sec:agent-model}. Fix $\mathit{goal} \in \mathbb{R}^2$ and $\theta \in \mathbb{R}$.
Let $\gamma : \mathbb{R}^4 \rightarrow \mathbb{R}^4$ and $\rho : \mathbb{R}^4 \rightarrow \mathbb{R}^4$ be defined as:
\begin{align}
\gamma(\state) &= [\state[0],\state[1] +\theta,(\state[2] - \mathit{goal}[0])\cos(\theta) + (\state[3] - \mathit{goal}[1])\sin(\theta),\nonumber\\
&\hspace{0.5in} -(\state[2] - \mathit{goal}[0])\sin(\theta)  + (\state[3] - \mathit{goal}[1])\cos(\theta)] \text{ and }\\
 \rho(p) &= [0,0, (p[2] - \mathit{goal}[0])\cos(\theta) + (p[2] - \mathit{goal}[1])\sin(\theta), \nonumber \\
&\hspace{0.5in} -(p[3] - \mathit{goal}[0])\sin(\theta)  + (p[3] - \mathit{goal}[1])\cos(\theta)].
\end{align}
 Then, for all $\state \in S$ and $p \in P$, $\gamma(f(\state,p)) = f(\gamma(\state), \rho(p))$, where $f$ is as in Section~\ref{sec:agent-model}. The transformation $\gamma$ would change the origin of the state space from $[0,0,0,0]$ to $[0,0,\mathit{goal}[0],\mathit{goal}[1]]$. Then, it would rotate the third and four axes counter-clockwise by $\theta$. Moreover, $\rho$ would set the first two coordinates of the parameters to zero as they do not affect the dynamics, translate the origin of the parameter space to $[0,0, \mathit{goal}[0],\mathit{goal}[1]]$, and rotate the third and fourth axes counter-clockwise by $\theta$. For the aircraft, this means translating and rotating the plane where the aircraft and the waypoint positions reside.
\end{example}

\subsection{Symmetry and reachtubes}
Computing reachtubes is computationally expensive. For example, one of the approaches for such computation entails simulating the system using an ODE solver, solving non-trivial optimization problems numerically to compute sensitivity functions, and computing the Minkowski sum of the solution with the sensitivity function \cite{C2E2paper,FanKJM16:Emsoft,FanM15:ATVA}. Another approach  entails Taylor approximations of the dynamics, integrations, and optimizations of the  time discretization of the solution \cite{CAS13,Chen2015ReachabilityAO}. It would also require several such computations to get tight enough reachtubes \cite{C2E2paper,Chen2015ReachabilityAO}. 
All of the these approaches face the curse of dimensionality: their complexity grows exponentially with respect to the dimension of the system.
Compared with that, transforming reachtubes is much cheaper, especially if the transformation is linear. For example, assume the tube $\rtube$ is represented using polytopes, i.e. for any pair $(X, [\tau_{i-1},\tau_i])$ in $\rtube$, $X$ is a polytope in $\mathbb{R}^n$. Linearly transforming $\rtube$ would require $\rtube.\ltube$ matrix multiplications. Hence, it would require polynomial time in $n$.

In our previous work \cite{Sibai:ATVA2019}, we showed how to get reachtubes of autonomous systems from previously computed ones using symmetry transformations. In this paper, we show how to do that for systems with parameters. This  allows different modes of a hybrid system and different agents with similar dynamics to share reachtube computations. That was not possible when the theory was limited to autonomous systems.

\begin{theorem}
	\label{thm:tube_trans_input}
If (\ref{sys:input}) is $\Gamma$-equivariant according to Definition~\ref{def:equivariance_input}, then for any $\gamma \in \Gamma$ and its corresponding $\rho$, any $K, p, [\ftime, \etime]$ and $\{(X_i,[\tau_{i-1},\tau_{i}])\}_{i=1}^{j}$ as a $(K,p,[\ftime,\etime])$-reachtube,
\[
\forall i \in [j], 
 \reachset(\gamma(K), \rho(p), [\tau_{i-1},\tau_{i}]) = \gamma (\reachset(K,p,[\tau_{i-1},\tau_{i}]))  \subseteq \gamma(X_i).
\]
\end{theorem}
\begin{proof}(Sketch)
The first part $ \reachset(\gamma(K), \rho(p), [\tau_{i-1},\tau_{i}]) = \gamma (\reachset(K,p, [\tau_{i-1},\tau_{i}])) $ follows directly from Theorem~\ref{thm:sol_transform_input_nonlinear}. The second part $\gamma (\reachset(K,p, [\tau_{i-1},\tau_{i}]))  \subseteq \gamma(X_i)$ follows from the reachtube $\Rtube(K,p, [\ftime, \etime]) $ is an over-approximation of the exact reachset during the small time intervals $[\tau_{i-1},\tau_{i}]$.
\end{proof}
Theorem~\ref{thm:tube_trans_input} says that we can transform a computed reachtube $\Rtube(K,p, [t_1, t_2]) = \{(X_i,[\tau_{i-1},\tau_{i}])\}_{i=1}^{j}$ to get another reachtube $\{(\gamma(X_i),[\tau_{i-1},\tau_{i}])\}_{i=1}^{j}$, which is an over-approximation of the reachsets starting from $\gamma(K)$.

Moreover, we present a corollary that is essential for using symmetry in the safety verification and synthesis algorithms of (\ref{sys:input}). 

\begin{corollary}
	\label{cor:tube_trans_input}
	Fix $\rtube = \Rtube(K,p,[\ftime,\etime]) = \{X_i, [\tau_{i-1},\tau_{i}]\}_{i=1}^{j}$ and $\rtube' = \Rtube(K',p',[\ftime,\etime]) = \{X_i', [\tau_{i-1},\tau_{i}]\}_{i=1}^{j}$, where $j = \rtube.\ltube$. If (\ref{sys:input}) is $\Gamma$-equivariant, and $K' \subseteq \mathbb{R}^n$, then if there exists $\gamma \in \Gamma$ and corresponding $\rho$ such that $K' \subseteq \gamma(K)$ and $p' = \rho(p)$, then
		$\forall i \in [j], X'_i \subseteq \gamma(X_i)$.
\end{corollary}
The corollary says that the transformation of the reachtube starting from $K$ with mode $p$ is an overapproximation of the one starting from $K'$ with mode $p'$ if $\gamma(K)$ contains $K'$.  

The results of this section subsume the results about transforming reachtubes of  dynamical systems without parameters presented in \cite{Sibai:ATVA2019}.

\section{Virtual system}
\label{sec:virtual}
The challenge in safety verification of multi-agent systems is that the dimensionality of the problem grows too rapidly with the number of agents to be handled by any of the current approaches. However, often agents share the same dynamics. 
For instance, several drones of the type described in Example~\ref{sec:single_drone_example} share the same dynamics but they may have different initial conditions and follow different waypoints. 
This commonality has been exploited in developing specialized proof techniques~\cite{JM:2012:small}. For reachability analysis, using symmetry transforms of the previous section, reachtubes of one agent in one mode can be used to get the reachtubes of other modes and even other agents.



Fix a particular value $p_v \in P$ and call it the {\em virtual} parameter. Assume that for all $p \in P$, there exists a pair of transformations $(\gamma_p, \rho_p)$ such that $\rho_p(p) = p_v$, $\gamma_p$ is invertible, and $\gamma_p (f (\state, p)) = f(\gamma_p(\state), \rho_p(p_v)) = f(\gamma_p(\state), p_v)$. 
Consider the resulting ODE:
\begin{align}
\label{sys:virtual}
	\frac{d \xi}{dt}(\statey, p_v, t) = f(\xi(\statey, p_v, t), p_v).
\end{align}
Following~\cite{russo2011symmetries}, we call (\ref{sys:virtual}) a {\em virtual system}. Correspondingly, we call (\ref{sys:input}), the {\em real system} for the rest of the paper. The virtual system unifies the behavior of all modes of the real system in one representative mode, the virtual one.

\begin{example}[Fixed-wing aircraft virtual system]
	\label{sec:virtual_example}
Consider the fixed-wing aircraft agent described in Example~\ref{sec:single_drone_example} and the corresponding transformation described in Example~\ref{sec:transformation_example}. Fix $p \in P$, we set $\mathit{goal}$ in the transformation of Example~\ref{sec:transformation_example} to $[p[2],p[3]]$ and $\theta$ to $\arctan_2(p[0] - p[2], p[3]-p[1])$ and let $\gamma_p$ and $\rho_p$ be the resulting transformations. Then, for all $p \in P$, $\rho_p(p) = [0,0,0,0]$. Hence, $p_v = [0,0,0,0]$ and the virtual system is that of Example~\ref{sec:single_drone_example} with the parameter $p = p_v$. 
For the aircraft, this means that $\gamma_p$ would translate the origin of the plane to the destination waypoint and rotate its axes so that the $y$-axis is aligned with the segment between the source and destination waypoints. Hence, in the constructed virtual system, the destination waypoint is the origin of the plane and the source waypoint is some-point along the $y$-axis although we fix it to the origin as well as it does not affect the dynamics.


\end{example}

 The solutions of the virtual system can be transformed to get solutions of all other modes in $P$ using $\{\gamma_p^{-1}\}_{p \in P}$ using Theorem~\ref{thm:sol_transform_input_nonlinear}.
\begin{theorem}
\label{thm:virtual_rep}
Given any initial state $\statey_0 \in S$, and any mode $p \in P$ and its corresponding $\gamma_p$ that maps it to the virtual system~(\ref{sys:virtual}), $\gamma_p^{-1}(\xi(\statey_0, p_v, \cdot))$ is a solution of the real system~(\ref{sys:input}) with mode $p$ starting from $\gamma_p^{-1}(\statey_0)$.
\end{theorem}
\begin{proof}
	Lets start with the first part of the theorem.
	Fix $p \in P$ and let $\state_0 = \gamma_p^{-1}(\statey_0)$. Using Theorem~\ref{thm:sol_transform_input_nonlinear}, $\gamma_p(\xi(\state_0,p,\cdot)) = \xi(\gamma_p(\state_0),\rho_p(p),\cdot))$ and is the solution of the real system~(\ref{sys:input}). Furthermore, $\rho_p(p) = p_v$ by definition and $\gamma_p(\state_0) = \gamma_p(\gamma_p^{-1}(\statey_0)) = \statey_0$. Hence, $\gamma_p(\xi(\state_0,p,\cdot)) = \xi(\statey_0,p_v,\cdot)$. Applying $\gamma_p^{-1}$ on both sides implies the first part of the theorem. 
\end{proof}

The following corollary extends the result of Theorem~\ref{thm:virtual_rep} to reachtubes. It follows from Theorem~\ref{thm:tube_trans_input}.

\begin{corollary}
	\label{cor:virtual_reachtube}
Given any initial set $K_v \in S$, and any mode $p \in P$, $\gamma_p^{-1}(\Rtube(K_v, p_v,[t_b,t_e]))$ is a reachtube of the real system~(\ref{sys:input}) with mode $p$ starting from $\gamma_p^{-1}(K_v)$. 
\end{corollary}

Consequently, we get a solution or a reachtube for each mode $p \in P$ of the real system by simply transforming a single solution or a single reachtube of the virtual system using the inverses of the transformations $\{\gamma_p\}_{p\in P}$. This will be the essential idea behind the savings in computation time of the new symmetry-based reachtube computation algorithm and symmetry-based safety verification algorithm presented next. It will be also the essential idea behind proving safety with unbounded time and unbounded number of modes.

\begin{example}[Fixed-wing aircraft infinite number of reachtubes resulting from transforming a single one]
	\label{sec:virtual_tube_example}
	Consider the real system in Example~\ref{sec:single_drone_example} and the virtual one in Example~\ref{sec:virtual_example}. Fix the initial set $K_r = [[1,\frac{\pi}{4}, 3,1], [2,\frac{\pi}{3}, 4,2]]$, $p_r = [2.5,0.5,13.3,5]$, and the time bound $20$ seconds. Then, using Example~\ref{sec:virtual_example}, $\theta = \arctan_2(2.5 - 13.3, 5 - 0.5) = -1.176$ rad and $\mathit{goal} = [13.3, 5]$. We call the corresponding transformation $\gamma_{p_r}$. Let $K_v = \gamma_{p_r}(K_r)$. Remember that $p_v = [0,0,0,0]$. Assume that we have the reachtube $\rtube_r = \Rtube(K_r,p_r,T)$. Then,  using Corollary~\ref{cor:virtual_reachtube}, we can get $\rtube_v = \Rtube(K_v, p_v, T)$ by transforming $\rtube_r$ using $\gamma_{p_r}$. The benefit of the corollary appears in the following: for any $p \in P = \mathbb{R}^4$, we can get $\Rtube(\gamma^{-1}_p(K_v),p,T)$ by transforming $\rtube_v$ using $\gamma_p^{-1}$.  

	The projection of $K_v$ on its last two coordinates $K_v \downarrow_{[2:3]}$ that represent the possible initial position of the aircraft in the plane would be a rotated square with angle $\theta$. The distance from $K_v \downarrow_{[2:3]}$ center to the origin would be equal to the distance from $K \downarrow_{[2:3]}$ center to the destination waypoint. Moreover, the angle between the $y$-axis and the center of $K_v \downarrow_{[2:3]}$ would be equal to the angle from the segment connecting the source and destination waypoints to the center of $K \downarrow_{[2:3]}$. On the other hand, $K_v \downarrow_{[1]} = K \downarrow_{[1]}$ and $K_v \downarrow_{[2]} = K \downarrow_{[2]} + \theta$. Now, consider any two waypoints, if the initial set of the aircraft relative position to the segment connecting the waypoints is the same as $K$, it is smaller or equal to $K$, and the time bound is less than or equal to $T$, we can get the corresponding reachtube by transforming $\rtube_v$ to the coordinate system defined by the new segment.
	
	In summary, the absolute positions of the aircraft and waypoints do not matter. What matters is their relative positions. The virtual system stores what matters and whenever a reachtube is needed for a new absolute position, it can be transformed from the virtual one.  
\end{example}

\section{Symmetry-based verification algorithm}
In this section, we introduce a novel safety verification algorithm, $\symsafety$, which uses existing reachability subroutines, but exploits symmetry. Unlike existing algorithms which compute the reachtubes of multi-agent systems from scratch for each agent and each mode, $\symsafety$ reuses the reachtubes computed from the virtual system. As we have discussed, transforming reachtubes is more  efficient than computing reachtubes. Therefore, $\symsafety$ can speed up  verification while preserving the soundness and precision of the reachability subroutines.
	In our earlier work~\cite{Sibai:ATVA2019}, we introduced  reachtube transformations using symmetry reduction for dynamical systems (single mode). Here, we extend the method across modes and we introduce the use of the virtual system, and we develop the corresponding verification  algorithm.

In Section~\ref{sec:tubecache_definition}, we define $\tubecache$---a data-structure for storing reachtubes; in~\ref{sec:symcompute}, we present the symmetry-based reachtube computation algorithm $\symcompute$ that resuses reachtubes stored in $\tubecache$; finally, in~\ref{sec:boundedtimesafety}, we define the $\safetycache$ data-structure which stores previously computed safety verification results which are used by the $\symsafety$ algorithm.

\subsection{$\tubecache$: shared memory for reachtubes}
\label{sec:tubecache_definition}
We show how we use the virtual system~(\ref{sys:virtual}) to create a shared memory for the different modes of the real system~(\ref{sys:input}) to reuse each others' computed reachtubes.
We call this shared memory $\tubecache$. It stores reachtubes, as the  over-approximation of reachsets, of the virtual system~(\ref{sys:virtual}).
	Later on, when dealing with the reachtube of the real systems, we will only need to transform the reachtubes in $\tubecache$ instead of computing the reachtubes again.

\begin{definition}
A $\tubecache$ is a data structure that stores a set of reachtubes of the virtual system~(\ref{sys:virtual}). It has two methods: $\mathtt{getTube}$, for retrieving stored tubes and $\mathtt{storeTube}$, for storing a newly computed one.
\end{definition}

 Given an initial set $K$, 
 	we want to avoid computing the reachtube for it for as large portion of $K$ as possible.  Therefore,
$\mathtt{getTube}$ would return a list of reachtubes 
$\ARtube(K_i, p_v, [0,T_i]), i \in [k]$, for some $k \in \mathbb{N}$ that are already stored in $\tubecache$. Moreover, the union of $K_i$s is the largest subset of $K$ that can be covered by the initial sets of the reachtubes in $\tubecache$.
 Formally,
\begin{align}
\label{def:tubecache}
\tubecache.\mathtt{getTube}(K) = \argmax_{\{\mathit{\ARtube(K_i,p_v, [0,T_i]) \in \tubecache\}_i}} \mathit{\text{Vol}(K \cap \cup_i {K_i})},
\end{align}
where $\mathit{\text{Vol}(\cdot)}$ is the Lebesgue measure of the set.
Note that for any $K \subset \mathbb{R}^n$, a maximizer of (\ref{def:tubecache}) would be the set of all reachtubes in $\tubecache$. However, this is very inefficient and the union of all reachtubes in $\tubecache$ would be too conservative to be useful for checking safety. Therefore, $\mathtt{getTube}$ should return the minimum number of reachtubes that maximize~(\ref{def:tubecache}). 

Note that the reachtubes  in $ \tubecache$ may have different time bounds. We will truncate or extend them when needed.
\subsection{$\symcompute$: symmetry-based reachtube computation}
\label{sec:symcompute}


Given an initial set $K \subset S$, a mode $p \in P$, and time bound $T > 0$, there are dozen of tools that tools that can return a $\ARtube(K,p,[0,T])$. See \cite{C2E2paper,CAS13,breachAD} for examples of such tools. We denote this procedure by  $\tubecompute(K,p,[0,T])$. 

	Instead of calling  $\mathtt{computeReachtube}(\cdot)$ whenever we need a reachtube, we want to use symmetry to reduce the load of computing new reachtubes by retrieving reachtubes that are already stored in $\tubecache$.
  We introduce Algorithm~\ref{code:sym_annotations} which implements this idea and denote it by $\symcompute$. 
  We map the initial set $K$ to an initial set $K_v$ using $\gamma_p$, i.e. $K_v = \gamma_p(K)$, where we use the subscript $v$ for ``virtual''. Then, we map the resulting reachtube from the call using $\gamma_p^{-1}$.


The input of $\symcompute$ consists of $K_v$ and $T$, and the $\tubecache$ that stores already computed reachtubes of the virtual system (\ref{sys:virtual}). 

First, it initializes $\mathit{restube}_v$ as an empty tube of the virtual system~(\ref{sys:virtual}) to store the result in line~\ref{ln:initialize_restube}. 
It then gets the reachtubes from $\tubecache$ that corresponds to $K_v$ using the $\mathtt{getTube}$ method in line~\ref{ln:gettube}. Now that it has the relevant tubes $\mathit{storedtubes}$, it adjusts their lengths based on the time bound $T$. For a retrieved tube with a time bound less than $T$ in line~\ref{ln:if_shorter}, $\symcompute$ extends the tube for the remaining time using $\tubecompute$ in lines~\ref{ln:extend_tube1}-\ref{ln:extend_tube2}, store the resulting tube in $\tubecache$ instead of the shorter one in line~\ref{ln:storetube}. If the retrieved tube is longer than $T$ (line~\ref{ln:if_longer}), it trims it in line~\ref{ln:trimtube}. However, we keep the long one in the $\tubecache$ to not lose a computation we already did. Then, the tube with the adjusted length is added to the result tube $\mathit{restube}_v$ in line~\ref{ln:unionres}.

The union of the initial sets of the tubes retrieved $\mathit{storedtubes}$ may not contain all of the initial set $K_v$. That uncovered part is called $K_v'$ in line~\ref{ln:remK}. The reachtube starting from $K_v'$ would be computed from scratch using $\tubecompute$ in line ~\ref{ln:compute_remtube}, stored in $\tubecache$ in line~\ref{ln:store_remtube}, and added to $\mathit{restube}_v$ in line~\ref{ln:union_remtube}. The resulting tube of the virtual system~(\ref{sys:virtual}) is returned in line~\ref{ln:return_tuberes}. This tube would be transformed by the calling algorithm using $\gamma_p^{-1}$ to get the corresponding tube of the real system~(\ref{sys:input}).

\begin{algorithm}
	\small
	\caption{$\symcompute$}
	\label{code:sym_annotations}
	\begin{algorithmic}[1]
		\State \textbf{input:} $K_v, T, \tubecache$
		\State $\mathit{restube}_v \gets \emptyset$ \label{ln:initialize_restube}
		\State $\mathit{storedtubes} \gets \tubecache.\mathtt{getTube}(K_v)$ \label{ln:gettube}
		\For{$i \in [|\mathit{storedtubes}|]$} \label{ln:for_loop_fixtubelength}
		\If{$\mathit{storedtubes}[i].T < T$} \label{ln:if_shorter}
		\State $(K_i,[\tau_{i-1}, T_i]) \gets \mathit{storedtubes}[i].\mathit{end}$ \label{ln:extend_tube1}
		\State $\mathit{tube}_i \gets \mathit{storedtubes}[i] \concat \tubecompute(K_i, p_v, [0,T - T_i])$ \label{ln:extend_tube2}
		\State $\tubecache.\mathtt{storeTube} (\mathit{tube}_i)$ \label{ln:storetube}
		\ElsIf{$\mathit{storedtubes}[i].T > T$}   \label{ln:if_longer}
		\State $\mathit{tube}_i \gets \mathit{storedtubes}[i].\mathit{truncate}(T)$ \label{ln:trimtube}
		\EndIf
		\State $\mathit{restube}_v \gets \mathit{restube}_v \cup tube_i$ \label{ln:unionres}
		\EndFor
		\State $K_{v}' \gets K_v \textbackslash \cup_i \mathit{storedtubes}[i].K$ \label{ln:remK}
		\State $\mathit{tube}' = \tubecompute(K_v', p_v, [0,T])$ \label{ln:compute_remtube}
		\State $\tubecache.\mathtt{storeTube} (\mathit{tube}')$  \label{ln:store_remtube}
		\State $\mathit{restube}_v \gets \mathit{restube}_v \cup \mathit{tube}'$ \label{ln:union_remtube}
		\State \textbf{return:} {$\mathit{restube}_v$} \label{ln:return_tuberes}
	\end{algorithmic}
\end{algorithm}

\begin{theorem}
	\label{thm:tube_comp_alg_sound}
The output of Algorithm~\ref{code:sym_annotations} is an over-approximation of the reachtube $\Rtube(K_v,p_v,[0,T])$.
\end{theorem}
\begin{proof}
$\tubecompute$ always return over-approximations of the reachset from a given initial set and time bound, and $\mathit{restube}_v$ contains reachtubes that have been computed by $\tubecompute$  at some point. 
There are three types of reachtubes in $\mathit{restube}_v$:
\begin{enumerate}
\item When the time bound $T_i$ of the stored reachtube $\mathit{storedtubes}[i]$ is less than $T$, we need to extend $\mathit{storedtubes}[i]$ until time $T$ by concatenating the original tube with the result of the call $\tubecompute(K_i, p_v, [0,T - T_i])$  where $[K_i, [\tau_{i-1},T_i]$ is the last pair in $\mathit{storedtubes}[i]$. It is easy to check that the concatenated two reachtubes is a valid $(\mathit{storedtubes}[i].K, p_v, [0,T])$- reachtube.
\item When the time bound $T_i$ of the stored reachtube $\mathit{storedtubes}[i]$ is more than $T$, the truncated reachtube is also a valid  $(\mathit{storedtubes}[i].K, p_v, [0,T])$- reachtube.
\item For $K'_v$ that is not contained in the union of the initial sets in $\mathit{storedtubes}$, $\tubecompute$  will return a valid $(K'_v, p_v, [0,T])$- reachtube.
 Finally, the union of reachtubes is a reachtube of the union of the initial sets.
\end{enumerate}
The union of the initial sets of the tubes in $\mathit{storedtubes}$ and $K_v'$ contains $K_v$, so the union of the reachtubes  the algorithm returns a $(K_v, p_v, [0,T])$- reachtube.
\end{proof}

The importance of $\symcompute$ lies in that if a mode $p$ required a computation of a reachtube and the result is saved in $\tubecache$, another mode with a similar scenario with respect to the virtual system would reuse that tube instead of computing one from scratch. Moreover, reachtubes of the same mode might be reused as well if the scenario was repeated again. 


\subsection{Bounded time safety }
\label{sec:boundedtimesafety}

In this section, we show how to use $\tubecache$ and $\symcompute$ of the previous section for bounded and unbounded time safety verification of the real system~(\ref{sys:input}). We consider a scenario where the safety verification of multiple modes of the real system~(\ref{sys:input}) starting from different initial sets and for different time horizons is needed. We will use the virtual system~(\ref{sys:virtual}) and the transformations $\{\gamma_p\}_{p \in P}$ to share safety computations across modes, initial sets, time horizons, and unsafe sets.

 We first introduce $\safetycache$, a shared memory to store the results of intersecting reachtubes of the virtual system~(\ref{sys:virtual}) with different unsafe sets. It will prevent repeating safety checking computations by different modes under similar scenarios. It will also play the key role to deduce unbounded time safety properties of the real system~(\ref{sys:input}) later in the section.

\begin{definition}
	A $\safetycache$ is a data structure that stores the results of intersecting reachtubes of the virtual system~(\ref{sys:virtual}) with unsafe sets. It has two functions: $\mathtt{getIntersect}$, for retrieving stored results and $\mathtt{storeIntersect}$, for storing a newly computed one.
\end{definition}

Given an initial set $K_v$, a time bound $T$, and an unsafe set $U_v$, the reachtube $\AReach(K_v, p_v, [0,T])$ is going to be unsafe if another reachtubes $\AReach(K'_v, p_v, [0,T'])$ which is already stored  in $\tubecache$ is unsafe, and $\AReach(K'_v, p_v, [0,T'])$ is an under-approximation of $\AReach(K_v, p_v, [0,T])$. Similarly, if $\AReach(K'_v, p_v, [0,T'])$ is an over-approximation of $\AReach(K_v, p_v, [0,T])$ and is safe, then  $\AReach(K_v, p_v, [0,T])$ is safe.
Formally, the $\mathtt{getIntersect}$ function of $\safetycache$ returns the truth value of the predicate $\ARtube(K_v,p_v,[0,T]) \cap U_v \neq \emptyset$ if a subsuming computation is stored, and returns $\bot$, otherwise. 

Formally, $\safetycache.\mathtt{getIntersect}(K_v,T,U_v) =$
\begin{align}
\begin{cases}
0, \text{ if } \exists K_v',T', U_v'\ |\ K_v \supseteq K_v', T\geq T', U_v \supseteq U_v', (K_v',T',U_v') \in \safetycache, \mathit{rtube} \cap U_v \neq \emptyset, \\
1, \text{ if } \exists K_v',T', U_v'\ |\ K_v \subseteq K_v', T\leq T', U_v \subseteq U_v', (K_v',T',U_v') \in \safetycache, \mathit{rtube} \cap U_v = \emptyset, \\
\bot, \text{ otherwise,}
\end{cases}
\end{align}
where $\mathit{rtube} = \ARtube(K_v',p_v, [0,T'])$, $0$ means $\unsafe$, $1$ means $\safe$, and $\bot$ means not stored.

It is equivalent to check the intersection of a reachtube of the real system~(\ref{sys:input}) with an unsafe set $U$ and to check the intersection of the corresponding mapped reachtube and unsafe set of the virtual one. This is formalized in the following lemma.

\begin{lemma}
	\label{lm:equiv_vir_real}
	Consider an unsafe set $U \subseteq \mathbb{R}^n \times \mathbb{R}^+$ and an over-approximated reachtube $\rtube = \ARtube(K,p,[t_1,t_2])$. Then, for any invertible $\gamma: \mathbb{R}^n \rightarrow \mathbb{R}^n$, $\mathit{rtube} \cap U \neq \emptyset$ if and only if $\gamma(\mathit{rtube}) \cap \gamma(U) \neq \emptyset$.
\end{lemma}
	
Now that we have established the equivalence of safety checking between the real and virtual systems, we present Algorithm~\ref{code:sym_safety} denoted by $\symsafety$. It uses $\safetycache$, $\tubecache$, and $\symcompute$ in order to share safety verification computations across modes. The method $\symsafety$ would be called several times to check safety of different scenarios and $\safetycache$ and $\tubecache$ would be maintained across calls.

$\symsafety$ takes as input an initial set $K$, a mode $p$, a time bound $T$, an unsafe set $U$, the transformation $\gamma_p$, and $\safetycache$ and $\tubecache$ that resulted from previous runs of the algorithm.

It starts by transforming the initial and unsafe sets $K$ and $U$ to a virtual system initial and unsafe sets $K_v$ and $U_v$ using $\gamma_p$ in line~\ref{ln:transformK_s}. It then checks if a subsuming result of the safety check for the tuple $(K_v,T,U_v)$ exists in $\safetycache$ using its method $\mathtt{getIntersect}$ in line~\ref{ln:getintersect}. If it does exist, it returns it directly in line~\ref{ln:resreturn}. Otherwise, the approximate reachtube is computed using $\symcompute$ in line~\ref{ln:symcomputetube}. The returned tube is intersected with $U_v$ in line~\ref{ln:intersecttubeunsafe} and the result of the intersection is stored in $\safetycache$ in line~\ref{ln:saveresult}. It is returned in line~\ref{ln:resreturn}.

\begin{algorithm}
	\small
	\caption{$\symsafety$}
	\label{code:sym_safety}
	\begin{algorithmic}[1]
		\State \textbf{input:} $K, p, T, U, \gamma_p, \safetycache, \tubecache$
		\State $K_v \gets \gamma_p(K)$, $U_v \gets \gamma_p(U)$ \label{ln:transformK_s} 
		\State $\mathit{result} \gets \safetycache.\mathtt{getIntersect}(K_v,T,U_v)$ \label{ln:getintersect}
		\If{$\mathit{result} = \bot$}
		\State $\rtube \gets \symcompute$ $(K_v, T,\tubecache)$ \label{ln:symcomputetube}
		\State $\mathit{result} \gets (\mathit{tube} \cap U_v = \emptyset)$ \label{ln:intersecttubeunsafe}
		\State $\safetycache.\mathtt{storeIntersect}(K_v,T,U_v, result)$ \label{ln:saveresult}
		\EndIf
		\State \textbf{return:} {$\mathit{result}$}\label{ln:resreturn}
	\end{algorithmic}
\end{algorithm}

The following theorem show $\symsafety$ is able to prove safety.
\begin{theorem}
 If $\symsafety$ returns $\safe$, then $\Rtube(K,p,[0,T]) \cap U = \emptyset$.
\end{theorem}
\begin{proof}
From Theorem~\ref{thm:tube_comp_alg_sound}, if the result is not stored in $\safetycache$, we know from Theorem~\ref{thm:tube_comp_alg_sound} that $\rtube$ in line~\ref{ln:symcomputetube} is an over-approximation of $\Rtube(K_v,p_v,[0,T])$. Moreover, we know from Corollary~\ref{cor:virtual_reachtube} that $\Rtube(K,p,[0,T]) \subseteq \gamma_p^{-1}(\rtube)$. But, from Lemma~\ref{lm:equiv_vir_real}, we know that the truth value of the predicate $(\rtube \cap U_v = \emptyset)$ is equal to that of $(\gamma_p^{-1}(\rtube) \cap U = \emptyset)$ and hence $\mathit{result}$ is $\safe$ if $\gamma_p^{-1}(\rtube) \cap U = \emptyset$ and thus it is $\safe$ if $\Rtube(K,p,T) \cap U = \emptyset$. Finally, the stored values in $\safetycache$ are results from previous runs, and hence have the same property. 
\end{proof}

However, if $\symsafety$ returns $\unsafe$, it might be that $\rtube$ in line~\ref{ln:symcomputetube} intersected the unsafe set because of an over-approximation error. There are two sources of such errors: the $\tubecompute$ method used by $\symcompute$ can itself result in over-approximation errors and actually it will most of the time~\cite{C2E2paper,CAS13}. But it may be exact too~\cite{bak2017hylaa}. The other source of errors is the $\tubecache.\mathtt{getTube}$ method which would return list of tubes with the union of their initial sets containing states that do not belong to the asked initial set. The first problem can be solved by: asking the method $\tubecompute$ to compute tighter reachtubes as existing methods provide this option at the expense of more computation~\cite{C2E2paper,CAS13}. We can use symmetry in these tightening computations as well, as we did in \cite{Sibai:ATVA2019}. We can also replace saved tubes in $\tubecache$ with newly computed tighter ones.
The second problem can be solved by asking $\tubecache.\mathtt{getTube}$ to return only the tubes with initial sets that are fully contained in the asked initial set. This would decrease the savings from transforming cached results, but it would reduce the false-positive error, saying $\unsafe$ while it is $\safe$.

\subsection{Unbounded time safety}

In this section, we show how can we deduce unbounded safety checks results from finite ones. 
The following corollary applies Lemma~\ref{lm:equiv_vir_real} to the transformations $\{\gamma_p\}_{p\in P}$ that map the different modes of the real system~(\ref{sys:input}) to the unique virtual one~(\ref{sys:virtual}).

\begin{corollary}[Infinite safety verification results from a single one]
	\label{cor:vir_real_eq_inter}
	Fix $U \subseteq \mathbb{R}^n$ and $\mathit{rtube} = \ARtube(K_v,p_v,[0,T]) $. If $\mathit{rtube} \cap U = \emptyset$, then for all $p \in P$, $\gamma_p^{-1}(\mathit{rtube}) \cap \gamma_p^{-1}(U) = \emptyset$.
\end{corollary}

The corollary means that from a single reachtube $\ARtube(K,p_v,[0,T])$ and unsafe set $U$ intersection operation, we can deduce the safety of any mode $p$ starting from $\gamma_p^{-1}(K)$ and running for $T$ time units with respect to $\gamma_p^{-1}(U)$. This would for example imply unbounded time safety of a hybrid automaton under the assumption that the unsafe sets of the modes are at the same relative position with respect to the reachtube. But, $\safetycache$ stores a number of results of such operations. We can infer from each one of them the safety of infinite scenarios in every mode. This is formalized in the following theorem which follows directly from Corollary~\ref{cor:vir_real_eq_inter}.

\begin{theorem}[Infinite safety verification results from finite ones]
\label{thm:unboundedsafety}
For any mode $p \in P$, initial set $K \subseteq S$, time bound $T \geq 0,$ and unsafe set $U \subset S \times \mathbb{R}^{\geq 0}$, such that $K \subseteq \gamma_p^{-1}(K')$, $U \subseteq \gamma_p^{-1}(U')$, and $\safetycache(K',T,U') = 1$, system (\ref{sys:input}) is safe. 
\end{theorem}

As more saved results are added to $\safetycache$, we can deduce the safety of more scenarios in all modes. If at a given point of time, we are sure that no new scenarios would appear, we can deduce the safety for unbounded time and unbounded number of agents with the same dynamics having scenarios already covered.

\begin{example}[Fixed-wing aircraft infinite number of safety verification results from computing a single one]
\label{sec:inf_safetychecks_example}
Consider the initial set $K$, mode $p$, time bound $T$, and their corresponding virtual ones $K_v$ and $p_v$ considered in Example~\ref{sec:virtual_tube_example}. Let the unsafe set be $U = [[0, -\infty, 11.9,5.1],[\infty, \infty, 12.9, 6.1]] \times \mathbb{R}^{\geq 0}$ and $U_v = \gamma_{p_r}(U)$, where $\gamma_{p_r}$ is as described in Example~\ref{sec:virtual_tube_example}. Assume that $\rtube_v \cap U_v = \emptyset$ and the result is stored in $\safetycache$. Then, for all $p \in P$, $\gamma_p^{-1}(\rtube_v) \cap \gamma_p^{-1}(U_v) = \emptyset$.

For the aircraft, $U$ could represent a building. Crashing with the building at any speed, heading angle, and time is unsafe. $U_v$ represents the relative position of the building with respect to the segment of waypoints.
Theorem~\ref{thm:unboundedsafety} says that for any initial set of states $K$ of the aircraft and time bound $T$, if the relative positions of the aircraft, unsafe set, and the segment of waypoints are the same or subsumed by those of $K_v$, $U_v$, and the origin, we can infer safety irrespective of their absolute positions.
\end{example}

\begin{example}[Fixed-wing aircraft unbounded time safety]
\label{sec:unbounded_example}
Consider a sequence of modes $ p_0 = [0,0,10,0], p_1 = [10,0,10,10]$, $p_3 = [10,10, 0,10]$, $p_4 = [0,10,0,20]$, $p_5 = [0,20,20,20], p_6 = [20,20,20,30] \dots $ in $P$ which draw a vertical concatenations of $S$ shaped paths with segments of size 10. Assume that the aircraft in Example~\ref{sec:single_drone_example} switches between different modes once its position $[\state[2], \state[3]]$ is in the square with side equal to $1$ around the destination waypoint $[p_i[2],p_i[3]]$, the absolute value of the angle with the segment is at most $\frac{\pi}{6}$, and its speed within $0.5$ of $v_c$. We call this the $\mathit{guard}(p_i)$. Moreover, we assume that for each waypoint, there exists a corresponding unsafe set $U_{p_i} = [[0, -\infty, b_1,b_2],[\infty,\infty,u_1,u_2]]$, where the square $[[b_1,b_2],[u_1,u_2]]$ has radius 1 and center 2 units to the left of the center of the segment of $p_i$. One would verify the safety of such a system by computing the reachtube $\rtube_{p_i}$ of a mode $p_i$ and checking if it is safe by intersecting with $U_{p_i}$. If it is safe, we compute the intersection of $\rtube_{p_i}$ with $\mathit{guard}(p_i)$ to get the initial set for $p_{i+1}$, and repeat the process again. However, we know that the initial set of any mode $p_i$ is at most $\mathit{guard}(p_i)$ which we know its relative position with respect to the segment. Moreover, we know the position of the unsafe set with respect to the segment. Then, one can deduce the unbounded-time safety of this system by just checking the safety of one segment as all scenarios are the same in the virtual system.

\end{example}

\section{Experimental evaluation}
\label{sec:experiments}

We implemented a safety verification software tool for multi-agent hybrid systems based on $\symcompute$ using Python 3. We call it $\ourtool$. We tested it on a linear dynamical system and the aircraft model of Example~\ref{sec:single_drone_example}  using DryVR~\cite{FanQMV:CAV2017} and Flow*~\cite{CAS13} as reachability subroutines.	

 The section is organized as follows: we describe the multi-agent verification algorithm $\multiagentverif$ that we designed and implemented in $\ourtool$ in Section~\ref{sec:multi-verif-algorithm}, describe the $\symcompute$ implementation in Section~\ref{sec:cacheimplementation}, 
 and finish with the results of the experiments and corresponding analysis in Section~\ref{sec:results}.

\subsection{$\ourtool$ and $\multiagentverif$: multi-agent safety verification algorithm}
\label{sec:multi-verif-algorithm}

The pseudo code for the multi-agent verification algorithm $\multiagentverif$ that we implement in $\ourtool$ is shown in Algorithm~\ref{code:multiagent-verif}. Our tool $\ourtool$ takes as input a JSON file specifying a list of $N$ agents of dimension $n$ in line~\ref{ln:multi_get_input}. It also specifies the python file that contains the dynamics function $f$ of Definition~\ref{def:model_syn} and the symmetry transformation function. The transformation $\gamma$ would be equal to the map $\gamma_p$ that transform solutions of the real system~(\ref{sys:input}) to those of the virtual one~(\ref{sys:virtual}), given a mode $p \in P$. The list of modes that the $\mathit{i}^{\mathit{th}}$ agent transition between sequentially and their corresponding guards is specified as well and denoted by $H_i$ in line~\ref{ln:multi_get_input}. The guard of the $j^{\mathit{th}}$ mode $H_i[j].p$ of the $i^{\mathit{th}}$ agent $H_i[j].\guard$ is a  predicate on the agent state which when satisfied, the agent transition mode to the $(j+1)^{\mathit{st}}$ mode. The guard $H_i[j]$ has time bound $H_i[j].T$ on how long the agent can stay in that mode. Moreover, for each agent $i \in [N]$, it specifies the initial set of states $K_i$. Finally, it specifies the static unsafe set $U$ and the subset $\relevant \subseteq [n]$ that is relevant for dynamic safety checking between agents. If the reachtubes of two agents projected on $\relevant$ intersect each other, it would model a collision between the agents. For example, $\relevant$ would be $\{2,3\}$ for the aircraft model in Example~\ref{sec:single_drone_example} as $(\state[2], \state[3])$ represents its position.
 
$\ourtool$ would return $\unsafe$ if the reachtubes of the agents starting from their initial sets of states and following the sequence of modes intersect static unsafe set in line~\ref{ln:multi_ret_static}, or when projected to $\relevant$, intersect each other in line~\ref{ln:multi_ret_dynamic}. It would return $\safe$ in line~\ref{ln:multi_safe}, otherwise. 

Currently, $\ourtool$ assumes that all agents share the same dynamics but do not interact. Hence, it has a single $\tubecache$ that is shared by all (line~\ref{ln:multi_init}). Current nonlinear hybrid-system verification tools such as C2E2~\cite{C2E2paper} or Flow*~\cite{CAS13} would consider the combination of the agents as a single system. The dimension of this system would be $n \times N$. It would have a parameter of size equal to $N$ times the size of the parameter vector of a single agent. This parameter would change values whenever an agent in the system switches modes. The unsafe set of such a system would be the static unsafe sets specified in the input JSON file and the dynamic unsafe set would be the safety relevant dimensions of a pair of agents being equal at the same time which detects a collision. This unsafe set would be time annotated.
Such verification tools would compute the reachtube of such a system sequentially on a per mode basis. The reachtube of the mode would be intersected with the guard to get the initial set of the next mode. Another reachtube would be computed starting from the computed initial set. Each of these tubes is intersected with the unsafe sets to get the verification result.

$\ourtool$ computes the reachtubes of individual agents independently. It would compute the reachtube $\modetube_{i}$ of the $j^{\mathit{th}}$ mode $H_i[j]$ of the $i^{\mathit{th}}$ agent in lines~\ref{ln:multi_transform_K}-\ref{ln:multi_transform_tube_back} using $\symcompute$ and intersect it with the guard using a function $\mathit{guardIntersect}$ to get the initial set $\initset_i$ for the next mode in line~\ref{ln:multi_guard_intersect}. In addition to $\initset_i$, $\mathit{guardIntersect}$ computes the minimum and maximum times: $\mathit{mintime}_i$ and $\mathit{maxtime}_i$, respectively, at which $\modetube_i$ intersects the guard. The value $\mathit{mintime}_i$ is the time at which a trajectory of the next mode may start at  and $\mathit{maxtime}_i$ is the maximum such time. These values will essential time information to check safety against time-annotated unsafe sets. That is because these values defined the range of time the trajectories starting from the initial set of the next mode may start at.

 The computed tube $\modetube_{i}$ gets appended to $\agenttube_i$ storing the full reachtube of the $i^{\mathit{th}}$ agent in line~\ref{ln:multi_transform_tube_back}.  The benefit of this method is that now all modes of all agents can be mapped to a single virtual system. They can resuse each others reachtubes using $\tubecache$  that is getting updated at every call to $\symcompute$. Moreover, the static safety is done in the usual way in line~\ref{ln:multi_ret_static}. The collision between agents is done by the function $\dynamicsafety$. 
 
 The function $\dynamicsafety$ takes two full reachtubes of two agents $\agenttube_1$ and $\agenttube_2$ along with two arrays $\mathit{lookback}_1$ and $\mathit{lookback}_2$. For agent $i$, the array $\mathit{lookback}_i$ consists of pairs of integers $(\mathit{ind}_j, \mathit{timerange}_j)$ specifying the index identifying the beginning of the $j^{\mathit{th}}$ mode tube in $\agenttube_i$ and the uncertainty in the starting time of the trajectories its initial set. $\dynamicsafety$ would use this information to time-align parts of $\agenttube_1$ and $\agenttube_2$ so that the intersection check happens only between two sets that may have been reached at the same time by the two agents.
 
 \begin{algorithm}
 	\small
 	\caption{$\multiagentverif$}
 	\label{code:multiagent-verif}
 	\begin{algorithmic}[1]
 		\State \textbf{input:} $\{K_i\}_{i=1}^N$, $\{H_i\}_{i=1}^N$, $\gamma$, $U$, $f$, $\relevant$ \label{ln:multi_get_input}
 		\State $\tubecache \gets \emptyset$, $\mathit{agentstubes} \gets \emptyset$, $\mathit{lookback} \gets \emptyset$ \label{ln:multi_init}
 		\For{agent $i$ in $[N]$}
 		\State $\initset_i \gets K_i$, $\agenttube_i \gets \emptyset$, $\mathit{lookback}[i] \gets \emptyset$ 
 		\For{mode index $j$ in $[\mathit{length}(H_i)]$ }
 		\State $K_{v,i} \gets \gamma_{i,j}(\initset_i)$ \label{ln:multi_transform_K}
 		\State $\modetube_{v,i} \gets \symcompute\mathit{(K_{v,i}, H_i[j].T,\tubecache)}$ \label{ln:multi_compute_tube}
 		\State $\modetube_i \gets \gamma_{i,j}^{-1}(\modetube_{v,i})$ \label{ln:multi_transform_tube_back}
 		\If{$\modetube_{i} \cap U$}
 		\State \textbf{return:} $\unsafe$ \label{ln:multi_ret_static}
 		\EndIf
 		\State $(\initset_i, \mathit{mintime}, \mathit{maxtime}) \gets \guardintersect(\modetube_i, H_i[j].\mathit{guard})$ \label{ln:multi_guard_intersect}
 		\State $\mathit{lookback}[i].\mathit{append}((\agenttube_i.\ltube, \mathit{maxtime} - \mathit{mintime}))$
 		\State $\agenttube_i \gets \agenttube_i \concat \modetube_i$ \label{ln:appendtube}
 		\EndFor
 		\For{agent $k$ in $[i-1]$}
 		\If{$\dynamicsafety\mathit{(\agenttube_i,\mathit{lookback}[i], \agenttube_k, \mathit{lookback}[k],\relevant)} == \unsafe$}
 		\State \textbf{return:} $\unsafe$ \label{ln:multi_ret_dynamic}
 		\EndIf
 		\EndFor
 		\State $\mathit{agentstubes}[i] \gets \agenttube_i$
 		\EndFor
 		\State \textbf{return:} $\safe$, $\mathit{agentstubes}$ \label{ln:multi_safe}
 	\end{algorithmic}
 \end{algorithm}


\subsection{$\symcompute$ implementation}
\label{sec:cacheimplementation}

To implement the $\tubecache$ in $\ourtool$, we grid the state space with a resolution $\delta \in \mathbb{R}^n$, where $\mathbb{R}^n$ is the state space of an agent. It would be then a dictionary with keys being centers of the grid cells and values being reachtubes with initial sets being the corresponding cells.

 Given an initial set $K_v$, $\symcompute$ would compute its bounding box, quantize its bottom-left and upper-right corners with respect to the grid, iterate over all cells and check if that cell intersects with $K_v$. If it does, we check if $\tubecache$ has the corresponding tube and add it to the result. Otherwise, we compute it from scratch and store it in $\tubecache$ with the key being the center of the cell. Then, $K_v'$ in $\symcompute$ would consist of the cells that it did not find the reachtube for and $\mathit{storedtubes}$ would consists of the reachtubes that we retrieved.

\subsection{Experimental results}
\label{sec:results}
We ran experiments using our tool $\ourtool$ on two models: a 3-dimensional linear dynamical system example and the nonlinear aircraft model described in Example~\ref{sec:single_drone_example}. The linear model is of the form $\dot{\state} = A(\state - p[3:6])$, where $A = [[-3,1,0],[0,-2,1],[0,0,-1]]$, $\state \in \mathbb{R}^3$, and  $p \in \mathbb{R}^6$. We considered scenarios with single, two, and three agents for each model following different sequences of waypoints. The sequences of waypoints for the linear model are translations and rotations of an $S$ shaped path. For the aircraft model, the paths are random crossing paths going north-east. In every scenario, all the agents have the same model. In the aircraft scenarios, the agent would switch to the next waypoint once its x, y position is within 0.5 units from the current waypoint in each dimension. The initial set of the aircraft was of size 1 in the position components, 0.1 in the speed, and 0.01 in the heading angle. We used Flow*~\cite{CAS13} and DryVR~\cite{FanQMV:CAV2017} to compute reachtubes from scratch for the linear example.  We only used DryVR for the aircraft model since our C++  Flow* wrapper does not handle a model having $\mathit{atan2}$ in the dynamics.  We ran Algorithm~\ref{code:multiagent-verif} on all scenarios in $\ourtool$ with and without using $\tubecache$. 
The symmetry used for the aircraft was the one we showed in Example~\ref{sec:virtual_example}. For the linear model, the symmetry transformation $\gamma_p$ that was used to map the state to the virtual system was by coordinate transformation where the new origin is at the next waypoint $p[3:5]$ and rotating the $xy$-plane by the angle between the previous and the next waypoints $p[0:2]$ and $p[3:5]$  projected to the plane.
We compared the computation time with and without symmetry and show the results in Table~\ref{table:experimental_results}. The reachtubes for a single and three linear agents are shown in Figures~\ref{fig:linear_nosym_1_agents_sym} and~\ref{fig:linear_nosym_3_agents_sym}.

\begin{table*}
	\centering
	\small
	\caption{\small Results.}\label{table:experimental_results}
\begin{tabular}{| l | l | l | l | l| l | l | l|}
	\hline 
	\multicolumn{2}{|l|}{tool \textbackslash $\:$ agent model}    &  \multicolumn{3}{|l|}{Linear(1,2,and 3 agents)} & \multicolumn{3}{|l|}{aircraft(1,2,and 3 agents)} \\
	\hline
	\multirow{3}*{\sf Sym-DryVR}      & computed &57 &90   &90 & 635.23&1181.38 &1550.62 \\
	& transformed &42 &165  &264 &20.76 &286.62 &501.38\\
	& time (min)  &0.093  &0.163   &0.187 & 3.42&8.2&10.59\\
	\hline 
	\multirow{3}*{\sf Sym-Flow*}      & computed &39.8 &61.14   &66.15 & && \\
	& transformed &19.2  &84.85  &143.85 &NA&NA&NA\\
	& time (min) &0.387  & 0.62  &0.684 &&& \\
	\hline 
	\multirow{2}*{\sf NoSym-DryVR}      & computed  &99 &255   &354 & 656&1468 &2052 \\
	& time (min)  &0.062  &0.355   &0.52 &3.71 &10.78 &15.47\\
	\hline 
	\multirow{2}*{\sf NoSym-Flow*}      & computed  &59 &151   &210 & && \\
	& time (min) &0.53  &1.328   &1.5 &NA&NA&NA \\
	\hline 
\end{tabular}
\end{table*}

\begin{figure}
	\begin{subfigure}[t]{0.5\textwidth}
		\includegraphics[width=\textwidth]{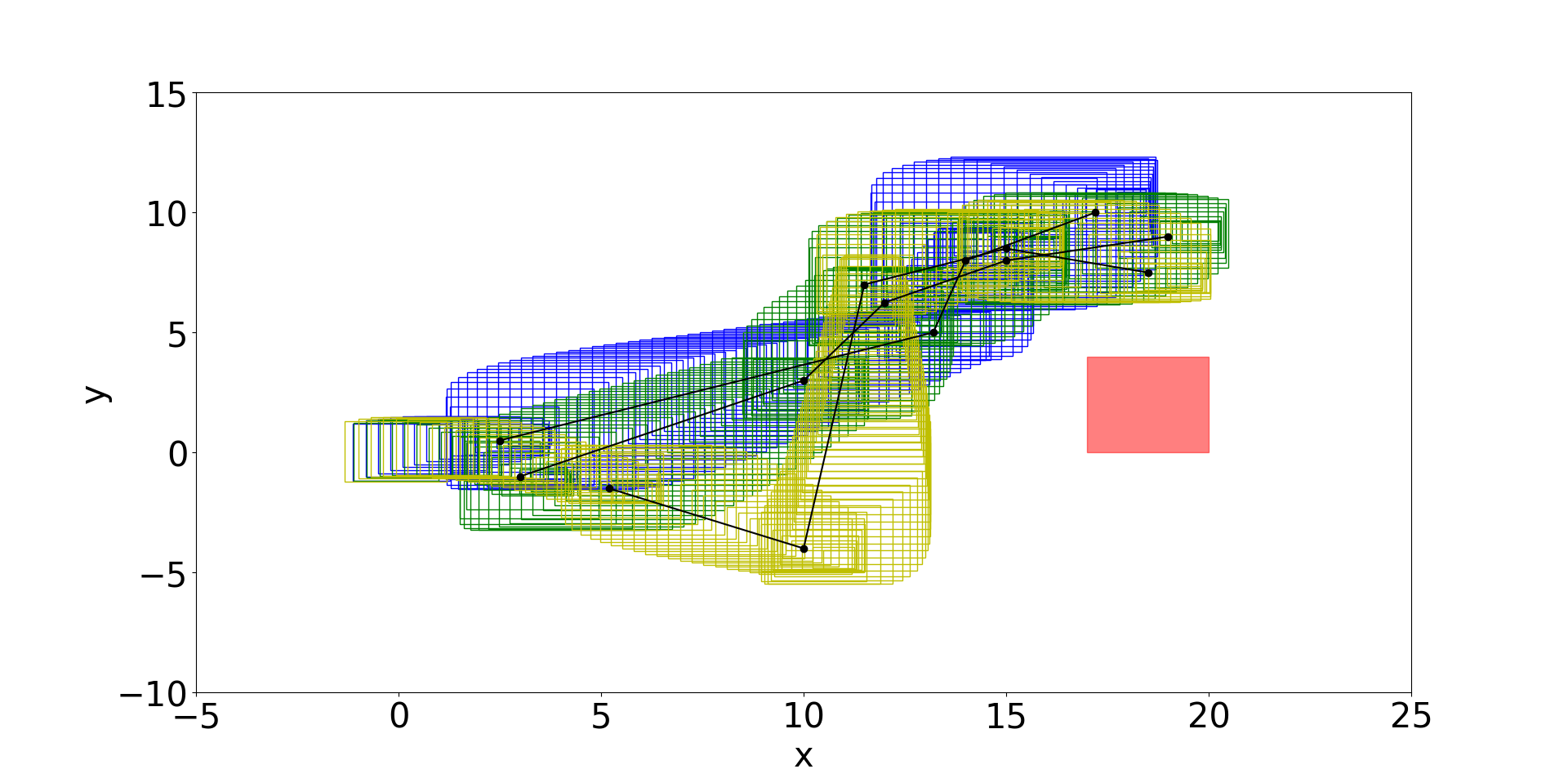}
		\caption{Reachtubes for 3 drones.\label{fig:linear_nosym_1_agents_sym}}
	\end{subfigure}
	\begin{subfigure}[t]{0.5\textwidth}
		\includegraphics[width=\textwidth]{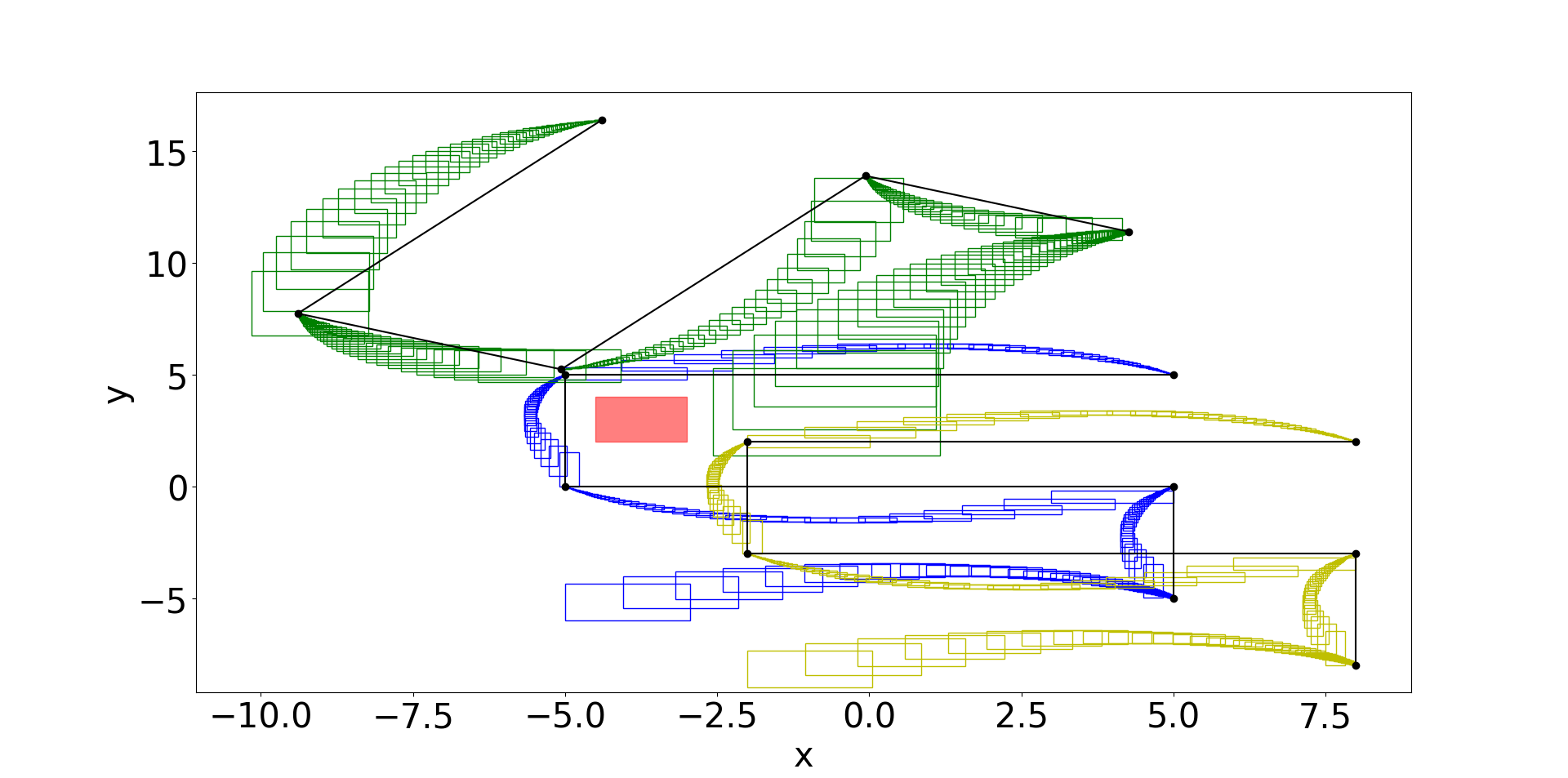} 	
		\caption{Reachtubes for 3 linear agents. \label{fig:linear_nosym_3_agents_sym}}

	\end{subfigure}
\end{figure}
In Table~\ref{table:experimental_results}, we call the combination of $\multiagentverif$ with DryVR while using $\tubecache$, {\sf Sym-DryVR} for symmetric DryVR. We call it {\sf Sym-Flow*} if we are using Flow* instead. If we are not using $\tubecache$, we call them {\sf NoSym-DryVR} and {\sf NoSym-Flow*}, respectively.
Remember in $\symcompute$, some tubes may be cached but they have shorter time horizons than the needed tube. So, we compute the rest from scratch. Here, we report the fractions of tubes computed from scratch and tubes that were transformed from cached ones. For example, if three quarters of a tube was cached and the rest was computed, we add 0.75 to the ``computed'' row and 0.25 to the ``transformed'' one.  Moreover, we report the execution time till the tubes are computed. In the experiments, we always compute the full tubes even if it was detected to be unsafe earlier to have a fair comparison of running times. Moreover, the execution time does not include dynamic safety checking as the four versions of the tool are doing the same computations. We are using $\ourtool$ in all scenarios with other reachability computation tools to decrease the degrees of freedom and show the benefits of transforming reachtubes over computing them. The {\sf Sym} versions result in decrease of running time up-to 2.7 times in the linear case with three agents. The ratio of transformed vs. computed tubes increases as the number of agents increase. This means that different agents are sharing reachtubes with each other in the virtual system. The total number of reachtubes is the same, whether $\tubecache$ is used or not. This means that the quality of the tubes , i.e. how tight they are, is the same whether we are transforming from $\tubecache$ or computing from scratch since the initial sets of modes are computed from intersections of reachtubes with guards. If the reachtube is coarser, the larger the initial set would get and the more reachtubes need to be computed.





\subsection{Discussion and conclusions}
\label{sec:conclusion}

In this paper, we investigated how symmetry transformations and caching can help achieve scalable, and possibly unbounded, verification of multi-agent systems.  
We developed a notion of {\em virtual system\/} which  define symmetry transformations for a broad class of hybrid and dynamical agent models visiting a waypoint sequences. 
Using virtual system, we present a prototype tool called $\ourtool$ that builds a cache of reachtubes for the transformed virtual system, in a way that is agnostic of the representation of the reachsets and the reachability analysis subroutine used. 
Our experimental evaluation show significant improvement in computation time on simple examples and increased savings as number of agents increase.

Several research directions are suggested by these results. We aim to develop stronger static analysis approaches for automatically checking the hypothesis of Theorem~\ref{thm:unboundedsafety} which allows unbounded time safety verification. Applying $\ourtool$ to analyze higher-dimensional agent models using other reachability subroutines would provide practical insights about this approach. Finally, exploiting symmetries for synthesis is another interesting direction.

%
%
%
%
%
%
\small
 \bibliographystyle{splncs04}
 \bibliography{sayan1,hussein}
 
 \normalsize
 
 \appendix

\end{document}